\documentclass[11pt]{article}
\usepackage{graphicx} % Required for inserting images
\usepackage[margin=1in]{geometry}
\usepackage{natbib}
\usepackage{tcolorbox}

\usepackage{xcolor}
\definecolor{darkblue}{RGB}{0, 0, 139} 
\definecolor{redorange}{RGB}{227, 66, 52}
\definecolor{emeraldgreen}{RGB}{0, 155, 119}

\usepackage{enumitem}
\setlist[itemize]{itemsep=0.5pt, topsep=1pt}
\setlist[enumerate]{noitemsep=0.5, topsep=1pt}

\AddToHook{cmd/appendix/before}{\crefalias{section}{appendix}}

\AddToHook{cmd/appendix/before}{\crefalias{section}{appendix}\crefalias{subsection}{appendix}}

\usepackage{titlesec}

\titlespacing*{\section}{0pt}{1.6ex plus 0.9ex minus .2ex}{1ex plus .2ex}

\titlespacing*{\subsection}{0pt}{1.4ex plus 0.9ex minus .2ex}{0.75ex plus .2ex}

\usepackage{xcolor}

\usepackage{booktabs}
\usepackage{bm}
\usepackage{nicematrix}
\usepackage{float}
\usepackage{dsfont}
\usepackage{xspace}
\usepackage{changepage}

\usepackage{tikz}
\usetikzlibrary{fit, backgrounds, positioning, shapes.geometric, decorations.pathreplacing, calligraphy}

\usepackage[ruled]{algorithm2e} % For algorithms

\SetAlFnt{\small}
\SetAlCapFnt{\small}
\SetAlCapNameFnt{\small}
\SetAlCapHSkip{0pt}
\IncMargin{-\parindent}
\usepackage{amsthm}

\makeatletter
\def\th@plain{%
  \thm@notefont{}%
  \itshape
  \thm@headfont{\scshape}%
}
\makeatother

\usepackage{amssymb}
\usepackage{mathrsfs}
\usepackage{color,soul}
\usepackage{mathtools}

\usepackage{multirow}
\usepackage{rotating}
\usepackage{makecell}
\usepackage{caption}

\usepackage{hyperref}
\usepackage{thmtools}
\usepackage{thm-restate}
\usepackage{enumitem}

\usepackage{cleveref}
\usepackage{graphicx} % Required for inserting images

\usepackage[]{color-edits}% Package initialization
\addauthor[Ram]{rd}{blue}
\addauthor[TODO]{todo}{red}

\newtheoremstyle{compact}
    {4pt}
    {4pt}
    {\itshape}
    {}
    {\bfseries}
    {.}
    {.5em}
    {}
\theoremstyle{compact}

\newtheorem{lemma}{Lemma}
\newtheorem*{lemma*}{Lemma}
\newtheorem{observation}{Observation}
\newtheorem{definition}{Definition}

\newtheorem{fact}{Fact}

\crefname{observation}{observation}{observations}
\Crefname{observation}{Observation}{Observations}
\crefname{enumi}{property}{properties}
\Crefname{enumi}{Property}{Properties}

\newcommand{\mbf}[1]{\mathbf{#1}}

\newcommand{\opt}{\textsc{OPT}}

\newcommand{\bl}{\textsc{BL}}

\newcommand{\reg}{\textsc{Reg}}
\newcommand{\bricklay}{\textsc{brick-laying}\xspace}
\newcommand{\base}{\textsc{Bases}}
\newcommand{\rank}{\textsc{Rank}}
\newcommand{\mmo}{\textsc{MMO}\xspace}

\newcommand{\1}[1]{\mathds{1}\left\{#1\right\}}

\newcommand{\floor}[1]{\left\lfloor #1 \right\rfloor}
\newcommand{\pp}[1]{\left( #1 \right)^+}

\newcommand{\wh}[1]{\widehat{#1}}
\newcommand{\wt}[1]{\widetilde{#1}}

\newcommand{\E}{\mathbb{E}}

\newcommand{\R}{\mathbb{R}}
\newcommand{\Z}{\mathbb{Z}}
\newcommand{\calA}{\mathcal{A}}

\newcommand{\calI}{\mathcal{I}}

\newcommand{\calM}{\mathcal{M}}

\newcommand{\calP}{\mathcal{P}}

\newcommand{\bell}{\bm{\ell}}
\newcommand{\define}{\coloneqq}

\title{Optimal Online Equitable Allocation with Indivisible Resources}
\author{
    Ramiro N. Deo-Campo Vuong\thanks{
        Cornell University.
        Email: ramdcv@cs.cornell.edu.
        The author thanks Siddhartha Banerjee, Robert Kleinberg, and Aditya Prasad for their editorial support, suggestions to approach \Cref{lem:mmo:exchange}, and miscellaneous guidance.
    }
}
\date{}

\begin{document}

\maketitle
\thispagestyle{empty}

\begin{abstract}
    Equitable allocation of indivisible goods to agents in online settings is an algorithmic primitive with applications for load balancing, network routing, online marketplaces, and multi-agent systems.
    We consider a general setting in which allocations are constrained to be bases of discrete polymatroids that arrive online.

    Our work demonstrates that a simple, myopic algorithm called \textsc{Brick-Laying}, which greedily minimizes the sum of squared loads on agents, achieves a universal and objective-free notion of optimality called \textit{majorization minimax-optimality} \cite{banerjee2026wf} for this setting.
    As a consequence, \textsc{Brick-Laying} simultaneously guarantees minimax optimal competitive ratios and regret for \emph{all} Schur-concave and Schur-convex objectives, and for any number of agents and resources (despite being agnostic to problem scale).

    Departing from popular primal-dual analysis, we employ majorization to compare allocations.
    We leverage the conjugates of integer partitions -- which act as a discrete dual to majorization -- to characterize worst-case instances for the \textsc{Brick-Laying} algorithm.
    Our approach reveals a novel structural connection between the geometry of partitions and online equitable allocation.
\end{abstract}
\newpage

\setcounter{page}{1}
\section{Introduction}
Equitable allocation of indivisible resources (goods/chores) to agents in online settings is a fundamental algorithmic primitive.
In such problems, resources arrive over time, and must be irrevocably assigned according to some specified combinatorial constraints.
Equity is measured in terms of some function of the allocations, and optimizing it is challenging because assignments are made without knowledge of future constraints.

In this work, we consider the following general variant of this problem: feasible allocations are described by a sequence of discrete polymatroid bases that arrive online.
We model worst-case settings by assuming these allocation constraints are designed by an adaptive adversary who has access to past resource assignments.
Below, we provide some special cases of our model:
\vspace{1pt}

\begin{adjustwidth}{1.5em}{0pt}
    \noindent{\bf Online $b$-Matching \cite{kalyanasundaram2000optimal}:}
    There are $n$ offline nodes.
    When online node $t$ of $m$ total online nodes arrives, it reveals a neighborhood $N_t\subseteq [n]$ and must be irrevocably assigned to a neighboring offline node.
    The goal is to maximize $\sum_{i\in[n]} \min\{b, \bell(i)\}$ where agent $i$ is assigned $\bell(i)$ online nodes and $b$ is the common capacity of offline nodes.
    \vspace{1pt}

    \noindent{\bf Online Batched $b$-Matching \cite{feng2025batching}:}
    This problem is the same as the above, except online nodes arrive together in batches instead of one at a time.
    All online nodes in a batch are irrevocably matched before the arrival of the next batch.
    \vspace{1pt}
    
    \noindent{\bf Online Network Load Balancing:} There are $n$ servers.
    In an online fashion for $h$ rounds, an integral capacity flow network is revealed with servers as sink nodes.
    An integral max flow from sources, representing requests, to servers is selected.
    The goal is to route requests to minimize the sum of server latencies, where the latency of a server routed $l$ total requests is $\frac{1}{2} l (l + 1)$.
\end{adjustwidth}
\vspace{1pt}

Despite the generality of polymatroid constraints and numerous different equity objectives one could consider, these settings admit a simple solution: \bricklay.
This algorithm, which greedily allocates resources to minimize the sum of squared agent loads, satisfies a universal notion of equitable allocation optimality that we call \textit{majorization minimax-optimality} (\mmo).
As a consequence of \mmo, \bricklay guarantees minimax optimal regret and competitive ratios for all reasonable equity-promoting objectives and any number of agents and resources.

The performance of online allocation algorithms is often measured via competitive ratios, defined on one or many objectives.
Despite the numerous celebrated results derived from competitive ratio analysis and objective design, these approaches have several shortcomings.
First, objective design risks over-stylization, which may lead to unintuitive and sensitive algorithms. 
Second, natural objectives are neglected if they do not admit a satisfactory bound for the artificial measure of competitive ratios.
For instance, the Nash Social Welfare\footnote{Another form of Nash Social Welfare with the same desirable properties is $f(\bell) = \prod_{i\in[n]} \bell(i)$.} objective $f(\bell) = (\prod_{i\in [n]}\bell(i))^{1/n}$, despite its desirable properties \cite{caragiannis2019unreasonable}, is avoided in worst-case allocation settings because it prohibits a constant competitive ratio without additional side information \cite{banerjee2022online}. 
We use the objective-free \mmo condition to analyze allocation algorithms.
This notion, developed by \cite{banerjee2026wf} for online allocation with divisible resources, uses the well-studied majorization preorder to compare the equity of various allocations.
Our elegant \mmo analysis implies sharp competitive ratio optimality guarantees for \bricklay under many equity-promoting objectives, suggesting that \bricklay is the most robust algorithm for worst-case online equitable allocation.

% have deterred the study of fundamental settings when they preclude satisfactory competitive ratio bounds \hl{not sure if this is a fair assessment}.
% Our analysis instead focuses on the objective-free \mmo property.
% This notion of optimality directly compares the equity of allocations using the \textit{majorization} preorder.
% To prove \bricklay exhibits the \mmo, we employ majorization and game theoretic analysis to show that deviation from \bricklay cannot lead to more equitable outcomes against an adaptive adversary.
% The \mmo condition is markedly stronger than optimal competitive ratio guarantees, as it implies optimal competitive ratios and regret bounds under any reasonable equity-promoting objective.
% We believe that our results provide a satisfactory notion of optimality in online allocation settings that previously yielded unsatisfactory competitive-ratio results.

\subsection{Setting}
An algorithm designer and an adversary compete in a two-player zero-sum game, defined over $n$ agents and $m$ resources.
The adversary and algorithm designer move in alternating fashion.
On round $t$ of the game, the adversary presents a \textbf{discrete polymatroid} $P_t \define \{\mbf{x}\in\Z_{\ge 0}^{n}\mid \mbf{x}(A) \le r_t(A), \forall A\in 2^{[n]}\}$, where $r_t:2^{[n]}\to\Z_{\ge 0}$ is the rank function defining the polymatroid and $\mbf{x}(A) \define\sum_{i\in A}\mbf{x}(A)$.
The rank function is normalized (i.e., $r_t(\emptyset) = 0$), monotone (i.e., $r_t(A)\subseteq r_t(B)$ for all $A\subseteq B$), and submodular (i.e., $r_t(A) + r_t(B) \ge r_t(A\cap B) + r_t(A\cup B)$ for all $A, B$).
In response to the adversary's polymatroid, the algorithm designer selects a basis of the polymatroid $\mbf{x}_t \in \base(P_t) \define \{\mbf{x}\in P_t\mid \mbf{x}([n]) = r_t([n])\}$ which represents an allocation of resources.
The rank function $r_t$ is specified via a value oracle that, given any $A\subseteq [n]$, returns $r_t(A)$.
The game continues until all resources are exhausted: the adversary must generate polymatroids satisfying $m = \sum_{t\in[h]} \rank(P_t)$, where $h$ is the total number of rounds and $\rank(P_t) \define r_t([n])$ is the rank of polymatroid $P_t$.
Equivalently, the adversary can instead present discrete convex games, from which the algorithm designer selects cores (see \Cref{app:sec:combo_constraint}).

Both players may adapt their strategy on round $t$ to the history of play from rounds $1,\dots,t-1$.
We use $\calP$ to denote an adversarial strategy that deterministically\footnote{All results extend in expectation under a randomized strategy $\calP$; we assume determinism for ease of exposition.} maps the history of play to a polymatroid to present to the algorithm designer.
Similarly, $\calA$ is a strategy available to the algorithm designer; it randomly maps a history of play and a presented polymatroid to a basis to select.
A strategy profile $(\calA, \calP)$ realizes a possibly randomized sequence of polymatroids $(P_t)_{t\in[h]}$ and selected bases $(\mbf{x}_t)_{t\in[h]}$.
The profile yields an allocation (or \emph{load vector}) $\bell(\calA, \calP) \define \sum_{t\in[h]} \mbf{x}_t$; the set of feasible hindsight load vectors is $H(\calA, \calP) \define \{\sum_{t\in[h]}\mbf{z}_t \mid \mbf{z}_t\in \base(P_t)\}$.
Finally, $\calP_{n,m}$ denotes the collection of adversarial strategies on $n$ agents and $m$ resources.

The goal of the algorithm designer is to produce an allocation $\ell(\calA, \calP)$ that maximizes the output of some given equity measure $f:\Z^{n}_{\ge 0}\to \R_{\ge 0}$ (or minimizes some inequity measure $g:\Z^{n}_{\ge 0}\to \R_{\ge 0}$).
We compare the equity of the algorithm designer's load vector to that of the most equitable load vector in hindsight, scaled by a comparison factor $\alpha \ge 0$.
The cost to the algorithm designer when she plays $\calA$ against an adversary playing $\calP$ under equity measure $f$ with comparison factor $\alpha$ is:
\begin{align*}
    c_{f,\alpha}(\calA, \calP) \define \E_\calA\left[\alpha \cdot \sup_{\bell^*\in H(\calA, \calP)} f(\bell^*) - f(\bell(\calA, \calP))\right]
\end{align*}
We adopt the perspective of an algorithm designer seeking to develop a robust strategy against every adversary.
In particular, we aim to minimize the algorithm designer's $\alpha$-regret, or her greatest cost incurred against any adversary: $\reg_{f, \alpha}^{n,m}(\calA) \define \sup_{\calP\in \calP_{n,m}} c_{f, \alpha}(\calA, \calP)$.
The cost of an algorithm designer minimizing inequity measure $g$ is $c_{g,\alpha}(\calA, \calP) \define \E_\calA\left[g(\bell(\calA, \calP)) - \alpha \cdot\inf_{\bell^*\in H(\calA, \calP)} g(\bell^*) \right]$ and her regret is $\reg_{g,\alpha}^{n,m}(\calA) \define \sup_{\calP\in\calP_{n,m}} c_{g,\alpha}(\calA, \calP)$.
We will always associate objectives $f$ and $g$ with maximization and minimization, respectively.
Note that getting minimax optimal regret for \emph{all} $\alpha$ implies a minimax optimal competitive ratio.
%Formally, in the case of equity objective $f$, allocation strategy $\calA$ minimizes the competitive ratio $\comprat_f(\calA) \define \sup_{\calP\in\calP_{n,m}} \frac{\E_\calA f(\bell(\calA, \calP))}{\E_\calA\left[\sup_{\bell^*\in H(\calA, \calP)} f(\bell^*)\right]}$ if $\calA$ minimizes $\reg_{f,\alpha}^{n,m}(\calA)$ for all $\alpha \ge 0$. Competitive ratios can be similarly defined for inequity measures $g$.

\subsection{Majorization Minimax-Optimality}\label{subsec:intro:mmo}
In this section, we introduce majorization minimax-optimality (\mmo) as a natural objective-free notion of optimality and illuminate how it implies optimal $\alpha$-regret for all comparison factors $\alpha$.
To show that deterministic allocation strategy $\mathcal{A}$ minimizes regret under equity measure $f$ and comparison factor $\alpha$, it suffices to show $\reg_{f,\alpha}^{n,m}(\calA) \le \reg_{f,\alpha}^{n,m}(\calA')$ for all alternate policies $\mathcal{A}'$. Since regret considers worst-case adversarial play, this is equivalent to the following: for all deviation policies $\calA'$ and adversarial strategies $\calP \in\calP_{n,m}$, there exists an adversarial strategy $\mathcal{P}'\in \calP_{n,m}$ such that $\calA'$ is no better against $\calP'$ than $\mathcal{A}$ against $\mathcal{P}$:
\begin{align*}
    \forall \mathcal{A}' \,\, \forall \mathcal{P} \,\, \exists \mathcal{P}' \quad c_{f,\alpha}(\mathcal{A}, \mathcal{P}) \le c_{f,\alpha}(\mathcal{A}', \mathcal{P}').
\end{align*}
This condition still depends on an objective $f$.
To obtain an objective-free criterion, we shift to pairwise comparison of load vectors. 
Let $\opt_f(\calA, \calP) \in \arg\max_{\bell^*\in H(\calA, \calP)} f(\bell^*)$ denote an arbitrary vector that is most equitable in hindsight. 
A sufficient condition for the inequality above is that $(\calA, \calP)$ simultaneously achieves weakly higher allocation equity and weakly lower hindsight equity than $(\calA', \calP')$, i.e.,  $f(\bell(\calA, \calP)) \ge \E f(\bell(\calA', \calP'))$ and $f(\opt_f(\calA, \calP))\le \E f(\opt_f(\calA', \calP'))$.
This reduces optimality to pairwise comparisons of load vector equity, as measured by $f$.
To remove dependence on $f$, we can ask: is there an objective-free notion of vector equity that recovers the above as a consequence for all reasonable equity-promoting objectives?

Majorization theory answers this question in the positive.
Majorization, a preorder relation on vectors, is founded on the Dalton-Pigou transfer or Robin principle that specifies that transferring resources from the wealthy to the poor yields more equitable outcomes \cite{arnold1987majorization,marshall2011majorization}.
\begin{definition}[Majorization]\label{def:intro:maj}
    For $\mbf{x},\mbf{y}\in\Z_{\ge 0}^n$ with $\mbf{x}([n]) = \mbf{y}([n])$, vector $\mbf{x}$ majorizes $\mbf{y}$ (written as $\mbf{x}\succeq \mbf{y}$) when $\mbf{x}^{\uparrow}([k]) \le \mbf{y}^{\uparrow}([k]), \forall k\in[n]$, where the $i^{\text{th}}$ element of $\mbf{x}^{\uparrow}\in\Z_{\ge 0}^n$ is the $i^{\text{th}}$ smallest element of $\mbf{x}$.
\end{definition}
We equivalently write $\mbf{y}\preceq \mbf{x}$ when $\mbf{x}$ majorizes $\mbf{y}$.
If vectors $\mbf{x}$ and $\mbf{y}$ are equivalent up to permutation (i.e.,$\mbf{x} \preceq \mbf{y}$ and $\mbf{x}\succeq \mbf{y}$), we write $\mbf{x}\sim\mbf{y}$.
Notably, vectors lower in the majorization order are more equitable.
The majorization order is aligned with Schur-monotone functions, which include most equity-promoting objectives.
\begin{definition}[Schur Monotone Functions]
    Function $f:\Z_{\ge 0}^n \to \R$ is Schur-concave when $\mbf{x} \preceq \mbf{y}$ implies $f(\mbf{x}) \ge f(\mbf{y})$.
    In contrast, $g:\Z_{\ge 0}^n \to \R$ is Schur-convex when $\mbf{x} \preceq \mbf{y}$ implies $g(\mbf{x}) \le g(\mbf{y})$.
\end{definition}
A vector $\mbf{x}$ deemed to be most equitable under majorization maximizes Schur-concave functions, such as the matching objective $\sum_{i\in[n]} \min(b, \mbf{x}(i))$ for any $b > 0$, Nash Social Welfare (NSW), and egalitarian welfare.
Such an $\mbf{x}$ simultaneously minimizes Schur-convex objectives, which include popular inequity measures such as $p$-norms and makespan.
See \Cref{app:sec:maj} for more details.

Equipped with majorization, we now present majorization minimax-optimality.
We alter the previous sufficient condition for optimal regret in two ways.
First, we apply \Cref{fact:back:hind_maj_min} below, which establishes the existence of a majorization-minimal hindsight load vector.
This implies that there is a randomized load vector $\bell^*$, depending only on the realization of the set $H(\calA, \calP)$, such that $\bell^*\in \arg\max_{\bell\in H(\calA, \calP)} f(\bell)$ for all Schur-concave $f$.
We call this vector $\opt(\calA, \calP) \define \bell^*$.
Second, we use majorization to perform pairwise comparisons: $f(\mbf{x}) \le f(\mbf{y})$ is replaced by $\mbf{x} \succeq \mbf{y}$.
We require majorization to hold with probability $1$ to preserve the consequence of optimal $\alpha$-regret.

\begin{tcolorbox}
\begin{definition}[Majorization Minimax-Optimality]
    For $n$ agents and $m$ resources, a deterministic allocation strategy $\calA$ is majorization minimax-optimal (\mmo) if for all alternative strategies $\calA'$ and adversarial strategies $\calP\in \calP_{n,m}$, there exists an adversarial strategy $\calP'\in\calP_{n,m}$ satisfying the following with probability $1$:
    \begin{align*}
        \bell(\calA, \calP) \preceq \bell(\calA', \calP')
        \quad\quad\text{and}\quad\quad
        \opt(\calA,\calP) \succeq \opt(\calA', \calP')
    \end{align*}
\end{definition}
\end{tcolorbox}

Like an equilibrium strategy in the online allocation game, an \mmo strategy implies that deviation from the strategy leaves the algorithm designer susceptible to worse game outcomes.
The notable difference is that \mmo uses majorization to compare outcomes, rather than regret (or competitive ratios).
Thus, an \mmo strategy is more robust than a regular equilibrium strategy because it implies minimax optimal regret under any reasonable equity-promoting objectives.

\subsection{Brick-Laying}
At this point, it is unclear whether any strategy satisfies the stringent \mmo property.
A myopic allocation strategy, which we call \bricklay, is \mmo for any number of agents and resources.
Each round, this algorithm selects a majorization-minimal basis (i.e., a basis that minimizes the sum of squared loads on agents).
We call the algorithm \bricklay, as it allocates resources to agents with lower loads in a manner that resembles laying bricks to build a wall.

\begin{definition}[Brick-Laying]\label{def:intro:bl}
    Let $\mbf{x}_t,\dots,\mbf{x}_{t-1}$ be a history of allocations, and $\bell_{t-1} = \sum_{s\in[t-1]} \mbf{x}_s$ denote the intermediate load vector. The \bricklay algorithm selects any $\mbf{x}_t\in \base(P_t)$ such that $\bell_{t-1} + \mbf{x}_t \preceq \bell_{t-1} + \mbf{x}$ for all $\mbf{x}\in\base(P_t)$, as computed in \Cref{alg:back:bl}.
\end{definition}

\Cref{alg:back:bl}, which implements \bricklay, is an adaptation of the greedy algorithm from \cite{federgruen1986greedy}.
We provide proof of its guarantees (\Cref{prop:into:bl}) in \Cref{app:sec:bl}.

\begin{algorithm}[ht]
    \SetAlgoNoLine
    \KwIn{
        Value oracle access to rank function $r_t$ defining polymatroid $P_t$ and load vector $\bell_{t-1}$.
    }
    \KwOut{
        A basis $\mbf{x}_t\in\base(P)$ yielding majorization minimal $\mbf{x}_t + \bell_{t-1}$.
    }
    Construct rank function $\wt{r}_t(\cdot) \gets r_t(\cdot) + \bell_{t-1}(\cdot)$ defining polymatroid $\wt{P}_t$\;
    Initialize $\mbf{x}^{(0)} \gets \vec{0}$ and $k\gets 0$ \;
    \While{$\mbf{x}^{(k)}\notin \base(\wt{P}_t)$}{
        Select $i_k\in\arg\min\{\mbf{x}^{(k)}(i)\mid \mbf{x}^{(k)} + \chi_i \in \wt{P}_t\}$ for $i^{\text{th}}$ standard basis $\chi_i$, breaking ties arbitrarily\;
        Update $\mbf{x}^{(k+1)} \gets \mbf{x}^{(k)} + \chi_{i_k}$ then update $k\gets k+1$\;
    }
    \Return basis $\mbf{x}_t \gets \mbf{x}^{(k)} - \bell_{t-1}$\;
    \caption{\bricklay}
    \label{alg:back:bl}
\end{algorithm}

\begin{restatable}[\bricklay Majorization Minimality]{proposition}{blmajmin}\label{prop:into:bl}
    Given $P_t$ and $\bell_{t-1}$ as input, \Cref{alg:back:bl} computes $\mbf{x}_t\in\base(P_t)$ yielding majorization minimal $\mbf{x}_t + \bell_{t-1}$ in polynomial time.
\end{restatable}

The special case in which $P_t$ constrains the allocation of a single resource ($\rank(P_t) = 1$) is central to our analysis.
In such cases, \Cref{alg:back:bl} has a simple implementation; it returns a basis $\mbf{x}_t \gets \chi_i$, where $i$ is an agent with lowest load that can receive the $t^{\text{th}}$ resource: $i\in\arg\min\{\bell_{t-1}(i) \mid \chi_i\in\base(P_t)\}$.
We use this implementation throughout our analysis.

\Cref{prop:into:bl} establishes that polymatroids admit a majorization minimal basis.
Since polymatroids are closed under Minkowski sums \cite{schrijver2002combinatorial,fujishige2005submodular}, this implies the existence of a majorization minimal hindsight allocation (proof in \Cref{app:sec:bl}).

\begin{restatable}[Majorization Minimal Hindsight Allocation]{fact}{hindalloc}\label{fact:back:hind_maj_min}
    All realizations of the polymatroid sequence admitted by profile $(\calA, \calP)$ admit a majorization minimal hindsight allocation.
    We denote this vector $\opt(\calA, \calP)$, which is randomized according to the realized polymatroid sequence.    
\end{restatable}

\subsection{Main Result}

We can now formally state our main results, establishing that \bricklay is \mmo.

\begin{tcolorbox}
\begin{restatable}[\bricklay Optimality]{theorem}{blmmo}\label{thm:intro:mmo}
    The \bricklay algorithm is majorization minimax-optimal under any number of agents $n$ and resources $m$.
\end{restatable}
\end{tcolorbox}

As an immediate consequence of the \mmo property, we get that \bricklay simultaneously achieves optimal regret minimization under any Schur-monotone function, and any comparison factor $\alpha$ (and hence minimax optimal competitive ratios). 
%reasonable equity-promoting objective and comparison factor.

\begin{restatable}[\bricklay Regret]{corollary}{bloptregret}\label{cor:intro:bl_opt_regret}
    Fix parameters $n,m\in\Z_{\ge 0}^n$, and $\alpha > 0$.
    For any Schur-concave objective $f:\Z_{\ge 0}^{n} \to \R_{\ge 0}$ or Schur-convex objective $g: \Z_{\ge 0}^{n} \to \R_{\ge 0}$, the \bricklay strategy on $n$, $m$ yields lower $\alpha$-regret then any allocation strategy $\calA'$:
    \begin{align*}
        \reg_{f,\alpha}^{n,m}(\bl) \le \reg_{f, \alpha}^{n,m}(\calA')
        \quad\quad\text{or}\quad\quad
        \reg_{g,\alpha}^{n,m}(\bl) \le \reg_{g, \alpha}^{n,m}(\calA')
    \end{align*}
\end{restatable}

\noindent{\bf Discussion of Results: } Before proceeding, we highlight three aspects of this result. 
\begin{itemize}[nosep]
\item By establishing \bricklay is \mmo compared to \emph{any} algorithm facing an adaptive adversary,~\Cref{thm:intro:mmo} also implies that \bricklay is the optimal \emph{deterministic} policy against an \emph{oblivious} adversary -- one who designs the entire polymatroid sequence with knowledge of the algorithm.
%Against such adversaries, our results hold when comparing \bricklay to deterministic alternative strategies $\calA'$; 
It does not, however, apply when comparing to randomized strategies for oblivious adversaries, such as the \textsc{Ranking} algorithm for online matching~\cite{karp1990optimal}.

\item While allowing complex polymatroid constraints allows us to capture a wide variety of settings, interestingly, it does not give the adversary additional power:
in~\Cref{lem:mmo:nested_worst}, we show that the best response to \bricklay is an online semi-matching instance.
In this sense, there is no loss in considering polymatroids over more commonly studied semi-matching constraints.

\item Since polymatroids are closed under Minkowski sums, and \bricklay yields the hindsight optimal solution if given the full instance (\Cref{fact:back:hind_maj_min}), one can view~\Cref{thm:intro:mmo} as an online \emph{local-to-global} property: it shows that myopically using \bricklay on local views of the polymatroid is optimal in a minimax sense. 
\end{itemize}

\subsection{Related Work}\label{subsec:intro:related_work}
Majorization is a cornerstone for analyzing equity in many domains \cite{arnold1987majorization, hardy1988inequalities, marshall2011majorization}, including semi-matchings \cite{frank2022decreasing1, harvey2006semi}, network flows \cite{megiddo1974optimal}, load balancing \cite{bhargava2001using, kleinberg1999fairness}, and persuasion \cite{banerjee2024fair}.
It is linked to integer partition conjugates, which characterize feasible degree distributions in bipartite graphs and directed graphs via the Gale-Ryser and Fulkerson Theorems \cite{fulkerson1960zero,gale1957theorem, ryser1957combinatorial}.
Further, several equivalent characterizations exist for majorization-minimal polymatroid bases \cite{frank2018discrete, frank2022decreasing1}, which can be computed greedily \cite{federgruen1986greedy}.

\cite{karp1990optimal} pioneered online matching, with subsequent work simplifying its analysis \cite{birnbaum2008line,devanur2013randomized,eden2021economics} and extending matching to broader allocation problems \cite{kalyanasundaram2000optimal, mehta2007adwords, mehta2010online}.
Recent research has explored fractional variants \cite{buchbinder2009design, devanur2012online, hathcock2024online,patton2026online} and the use of predictions to optimize objectives, which otherwise admit poor competitive ratios \cite{banerjee2022online, chakrabarty2019approximation,kesselheim2023online}.
Most existing online allocation analyses rely on primal-dual techniques and differential equations.
However, as noted by \cite{devanur2012online}, these methods often insinuate a mysterious ``kind of duality between the algorithm and the counterexample''.
We suggest that in integral equitability settings, this duality is captured by the conjugate of integral partitions.

The applications of majorization to online equitable allocation remain sparse.
\cite{goel2005approximate} proves that a special case of \bricklay yields allocations that approximately minorize the optimal hindsight solution.
Most closely related to our work is \cite{banerjee2026wf}, which proves the \textsc{Water-Filling} algorithm is \mmo for online fractional semi-matching.
The techniques in~\cite{banerjee2026wf} critically depend on the uniqueness and functional form of the \textsc{Water-Filling} allocation, and it is unclear if the insights carry over to integer allocations.
To circumvent this, we establish novel structural connections between online allocation and integer partition conjugates, which then allow us to extend \mmo to more complex combinatorial constraints.

\section{Proof Overview and an Illustrative Example}\label{sec:tech_overview}
%The main technical contribution of this work is proving that \bricklay is an \mmo strategy (\Cref{thm:intro:mmo}).
Since \mmo resembles an equilibrium condition, our analysis parallels equilibrium analysis in two-player zero-sum games.
%We adopt the adversary's perspective and develop strategies that yield poor outcomes for the algorithm designer.
In particular, the crux of our proof is showing that the adversary's best response to the algorithm designer using \bricklay is to employ a \textit{nested strategy}.
\begin{definition}[Nested Strategy]
    An adversarial strategy $\calP\in \calP_{n,m}$ is nested if, for any strategy $\calA$, the profile $(\calA, \calP)$ induces a polymatroid sequence $P_1,\dots,P_m$ whose rank functions $r_1,\dots,r_m$ take the form $r_t(A) = \1{A\cap N_t\neq \emptyset}$ for $N_1\supseteq \dots\supseteq N_m$ with probability $1$. 
\end{definition}
Nested strategies pose a challenge because the algorithm designer would prefer to assign resources to agents who are absent from future neighborhoods, but she cannot identify such agents from past information.
Notably, nested strategies are semi-matching instances in which online nodes arrive one at a time: online node $t$ has neighborhood $N_t$.
% \Cref{lem:mmo:nested_worst} shows that polymatroids grant no additional power to the adversary beyond the special case of semi-matching.

In \Cref{lem:mmo:nested_worst}, we demonstrate that nested strategies are a best response to \bricklay.
We do so by treating the load vector produced by \bricklay as an integer partition of $m$ resources into $n$ parts (agents) and leveraging its conjugate.
\begin{definition}[Conjugate Load Vector]
    The conjugate load vector of $\bell\in \Z_{\ge 0}^{n}$ with $\bell([n]) = m$, denoted as $\bell^{\circ}\in\Z_{\ge 0}^{m}$, has $\bell^{\circ}(j) = |\{i\in [n]\mid \mbf{x}(i) \ge j\}|$ for all $j\in [m]$.
\end{definition}
Informally, the conjugate of a load vector -- that behaves like a discrete dual to majorization -- encodes the number of agents receiving at least $j$ resources for each $j$. In~\Cref{subsec:overview:illustration}, we illustrate via a simple example how we use the conjugate to establish that nested instances are a best response to \bricklay. The formal result for general polymatroids is presented in~\Cref{sec:mmo}.

The second part of our proof, \Cref{lem:mmo:vs_alt}, compares an alternative algorithm $\calA'$ to \bricklay.
Using~\Cref{lem:mmo:nested_worst}, we can focus only on nested strategies $\calP$. Our goal now is to construct $\calP'$ on which $(\bl, \calP)$ provides weakly better outcomes than $(\calA', \calP')$.
To construct such a $\calP'$, we relabel offline nodes so that agents with low load under strategy $\calA'$ never have another opportunity to be allocated a resource.
We note that a similar proof is performed in~\cite{banerjee2026wf} for fractional allocation problems.
Composing \Cref{lem:mmo:nested_worst} and \Cref{lem:mmo:vs_alt} shows that \bricklay is \mmo.

\subsection{Nestedness is Worst-Case: a Semi-Matching Example}\label{subsec:overview:illustration}
We present an online semi-matching example to give intuition for our most technical contribution: \Cref{lem:mmo:nested_worst}.
To show nested strategies are the adversary's best response to \bricklay, we analyze \Cref{alg:mmo:nest}, which takes as input an adversarial sequence $\calP$ and produces a nested strategy $\wh{\calP}$ on which \bricklay yields weakly worse outcomes.
To simplify our analysis, we assume \bricklay tie-breaks in favor of agents with a lower index when presented with a rank $1$ polymatroid.
This is without loss of generality because different tie-breaking yields the same intermediate load vector up to permutation; the adversary can adjust to these different permutations using their adaptive power.
Again, the $i^{\text{th}}$ standard basis is $\chi_i$ with $\chi_i(j) = \1{i = j}$.

\begin{algorithm}[ht]
    \SetAlgoNoLine
    \KwIn{
        Adversarial strategy $\calP\in\calP_{n,m}$.
    }
    \KwOut{
        $\wh{\calP}\in\calP_{n,m}$ nested under $\bl$ with $\bell(\bl, \calP) \sim \bell(\bl, \wh{\calP})$ and $\opt(\bl, \calP) \succeq \opt(\bl, \wh{\calP})$.
    }
    Define the empty sequence $S$\;
    Compute $\bell \gets \bell(\bl, \calP)$ and take $\bell^{\circ}(0) \gets n$\;
    \For{each epoch $j\in [m]$}{
        Append $\bell^{\circ}(j)$ copies of $\wh{P}$ with rank function $\wh{r}(A) = \1{A \cap [\bell^{\circ}(j-1)] \neq\emptyset}$ to the end of sequence $S$\;
    }
    \Return the policy $\wh{\calP}$ which deterministically plays polymatroids in $S$ regardless of history\;
    \caption{Nested is Worst-Case}
    \label{alg:mmo:nest}
\end{algorithm}
We note that each epoch $j$ can instead consist of a single polymatroid $P$ with rank function $r(A) = \bell^{\circ}(j) \cdot \1{A\cap [\bell^{\circ}(j-1)] \neq\emptyset}$; all our analysis holds with straightforward adjustments.

\begin{restatable}[Nested is Worst-Case]{lemma}{nested}\label{lem:mmo:nested_worst}
    Given $\calP\in\calP_{n,m}$, \Cref{alg:mmo:nest} yields strategy $\wh{\calP}\in \calP_{n,m}$, which is nested and satisfies: $\bell(\bl, \calP) \sim \bell(\bl, \wh{\calP})$ and $\opt(\bl, \calP) \succeq \opt(\bl, \wh{\calP})$.
\end{restatable}

The intuition behind~\Cref{alg:mmo:nest} is that we can exploit the conjugate of $\bell(\bl, \calP)$ to get an instance $\wh{\calP}$ on which the \bricklay allocation is preserved, while ensuring the hindsight solution $\opt(\bl, \calP)$ is feasible (up to permutation) in $\wh{\calP}$.
This immediately establishes~\Cref{lem:mmo:nested_worst}, as by majorization minimality we have $\opt(\bl, \calP) \succeq \opt(\bl, \wh{\calP})$.
%The \bricklay allocation $\bell\define \bell(\bl, \calP)$, encodes useful information in its conjugate.
To see why the conjugate helps, notice that $\bell^{\circ}(1)$ specifies how many agents must receive at least one resource.
Thus, to construct an epoch in which $\bell^{\circ}(1)$ agents get one resource while allowing for all possible hindsight allocations, we bring in $\bell^{\circ}(1)$ copies of a polymatroid $\wh{P}$ for which all indicators are a basis $\base(\wh{P}) = \{\chi_i\}_{i\in[n]}$.
The \bricklay allocation gives one resource to $\bell^{\circ}(1)$ unique agents (breaking ties in favor of agents with smaller indices by assumption); simultaneously, every allocation of the $\bell^{\circ}(1)$ resources to agents is feasible in hindsight.
This constitutes epoch $j = 1$.
All subsequent epochs $j$ abide by similar intuition, with a caveat:
agents in $[\bell^{\circ}(j-1)]$ have load $j-1$ and all other agents have load less than $j-1$.
To prevent low-load agents from getting allocations, we only allow allocations to agents in $[\bell^{\circ}(j-1)]$ in future rounds. 
This ensures our constructed instance induces the same \bricklay allocation as $\calP$, while allowing the largest set of hindsight allocations.

In~\Cref{fig:mmo:example1,fig:mmo:example2}, we illustrate our construction on an instance of online batched semi-matching from~\cite{harvey2006semi}.
In this setting, the $t^{\text{th}}$ polymatroid represents a batch $B_t\subseteq [m]$ of resources.
Each agent $i$ has a set of compatible resources $N_{i,t}\subseteq B_t$ that may be allocated.
The rank function $r_t$ of the polymatroid $P_t$ encoding every compatibility respecting allocation is a coverage function: $r(A) = \left|\bigcup_{i\in A} N_{i,t} \right|$.
We consider a strategy $\calP$ with $n=3$ agents, $m=5$ resources, and $2$ batches illustrated in \Cref{fig:mmo:example1}.
Neighborhoods $N_{i,t}$ are shown using edges from $i$ to resources in $N_{i,t}$.

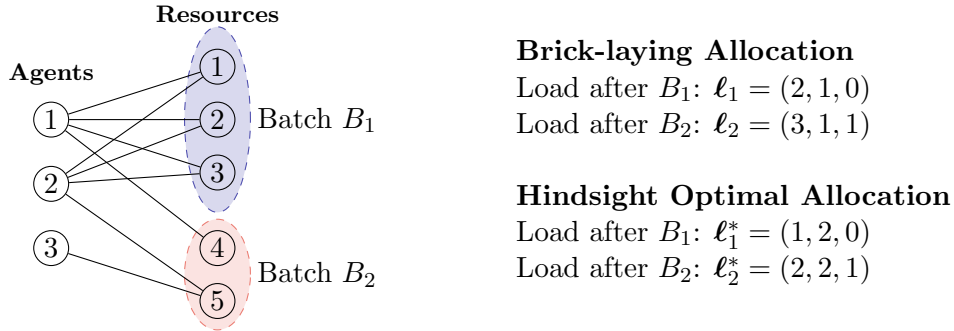
\begin{figure}[ht]
    \centering
    \hspace{0.05\textwidth}
    \begin{minipage}[c]{0.35\textwidth}
        \centering
        \vspace{-25pt}
         \begin{tikzpicture}[
            every node/.style={circle, draw, minimum size=2mm, inner sep=1.5pt},
            node distance=12mm,
            scale=0.4,
            batch_oval1/.style={
                ellipse, 
                draw=darkblue!70, 
                fill=darkblue!15, 
                dashed, 
                inner xsep=2pt, % Extra horizontal padding
                inner ysep=-2pt   % Extra vertical padding
            },
            batch_oval2/.style={
                ellipse, 
                draw=redorange!70, 
                fill=redorange!15, 
                dashed, 
                inner xsep=2pt, % Extra horizontal padding
                inner ysep=-2pt   % Extra vertical padding
            }
        ]
    
            \node[draw=none] at (0,7.5) {\textbf{\footnotesize Agents}};
            \node[draw=none] at (5.5,9.5) {\textbf{\footnotesize Resources}};

            % Agents (V)
            \node (v1) at (0,6) {1};
            \node (v2) at (0,3.875) {2};
            \node (v3) at (0,1.75) {3};
            
            % Resources (U)
            \node (u1) at (5.5,7.75) {1};
            \node (u2) at (5.5,6) {2};
            \node (u3) at (5.5,4.25) {3};
            \node (u4) at (5.5,1.75) {4};
            \node (u5) at (5.5,0) {5};
            
            % Edges
            \draw (v1) -- (u1); \draw (v1) -- (u2); \draw (v1) -- (u3); \draw (v1) -- (u4);
            \draw (v2) -- (u1); \draw (v2) -- (u2); \draw (v2) -- (u3); \draw (v2) -- (u5);
            \draw (v3) -- (u5);

            % Batches
            \begin{scope}[on background layer]
                \node[batch_oval1, fit=(u1) (u2) (u3), label=right:Batch $B_1$] (batch1) {};
            \end{scope}
            \begin{scope}[on background layer]
                \node[batch_oval2, fit=(u4) (u5), label=right:Batch $B_2$] (batch2) {};
            \end{scope}
        \end{tikzpicture}
    \end{minipage}
    \hfill
    \begin{minipage}[c]{0.5\textwidth}        
        {\bf Brick-laying Allocation} \\
        Load after $B_1$: $\bell_1 = (2, 1, 0)$ \\
        Load after $B_2$: $\bell_2 = (3, 1, 1)$ \\
        
        {\bf Hindsight Optimal Allocation} \\
        Load after $B_1$: $\bell_1^* = (1, 2, 0)$ \\
        Load after $B_2$: $\bell_2^* = (2, 2, 1)$ \\
    \end{minipage}
    \vspace{-10pt}
    \caption{\small\em
        The graph on the left shows an instance $\calP$ of online batched semi-matching with $n=3$ agents and $m=5$ resources, arriving over two rounds.
        %Upon the arrival of a batch, nodes in the batch must be irrevocably allocated to neighboring agents.
        On the right, we give the intermediate loads of \bricklay and $\opt(\bl, \calP)$, where $\bell_t(i)$ is the load on agent $i$ after round $t$.
        In round $1$, \bricklay is indifferent between loads $(2,1,0)$ and $(1,2,0)$;
        in round 2, the only majorization-minimal allocation is $(1,0,1)$, yielding final load $(3,1,1)$.
        $\opt(\bl, \calP) = (2,2,1)$ is the majorization minimal load vector over the entire instance.
    }\label{fig:mmo:example1}
\end{figure}

\begin{figure}[ht]
    \centering
    \hspace{0.05\textwidth}
    \begin{minipage}[c]{0.35\textwidth}
        \centering
        \vspace{-25pt}
         \begin{tikzpicture}[
            every node/.style={circle, draw, minimum size=2mm, inner sep=1.5pt},
            node distance=12mm,
            scale=0.4,
            epoch_oval1/.style={
                ellipse, 
                draw=darkblue!70, 
                fill=darkblue!15, 
                dashed, 
                inner xsep=2pt, % Extra horizontal padding
                inner ysep=-2pt   % Extra vertical padding
            },
            epoch_oval2/.style={
                ellipse, 
                draw=redorange!70, 
                fill=redorange!15, 
                dashed, 
                inner xsep=0pt, % Extra horizontal padding
                inner ysep=0pt   % Extra vertical padding
            },
            epoch_oval3/.style={
                ellipse, 
                draw=emeraldgreen!70, 
                fill=emeraldgreen!15, 
                dashed, 
                inner xsep=0pt, % Extra horizontal padding
                inner ysep=0pt   % Extra vertical padding
            }
        ]
    
            \node[draw=none] at (0,7.5) {\textbf{\footnotesize Agents}};
            \node[draw=none] at (5.5,9.5) {\textbf{\footnotesize Resources}};

            % Agents (V)
            \node (v1) at (0,6) {1};
            \node (v2) at (0,3.875) {2};
            \node (v3) at (0,1.75) {3};
            
            % Resources (U)
            \node (u1) at (5.5,7.75) {1};
            \node (u2) at (5.5,6) {2};
            \node (u3) at (5.5,4.25) {3};
            \node (u4) at (5.5,1.75) {4};
            \node (u5) at (5.5,-0.25) {5};
            
            % Edges
            \draw (v1) -- (u1); \draw (v1) -- (u2); \draw (v1) -- (u3); \draw (v1) -- (u4); \draw (v1) -- (u5);
            \draw (v2) -- (u1); \draw (v2) -- (u2); \draw (v2) -- (u3); \draw (v2) -- (u4);
            \draw (v3) -- (u1); \draw (v3) -- (u2); \draw (v3) -- (u3); \draw (v3) -- (u4);

            \draw[thick, decoration={brace, raise=5pt}, decorate]
                (u1.north east) -- (u3.south east)
                node[midway, xshift=31pt, draw=none] {Epoch 1};

            \draw[thick, decoration={brace, raise=5pt}, decorate]
                (u4.north east) -- (u4.south east)
                node[midway, xshift=31pt, draw=none] {Epoch 2};

            \draw[thick, decoration={brace, raise=5pt}, decorate]
                (u5.north east) -- (u5.south east)
                node[midway, xshift=31pt, draw=none] {Epoch 3};
        \end{tikzpicture}
    \end{minipage}
    \hfill
    \begin{minipage}[c]{0.5\textwidth}        
        {\bf Brick-laying Allocation} \\
        Load after epoch 1: $\bell_3 = (1, 1, 1)$ \\
        Load after epoch 2: $\bell_4 = (2, 1, 1)$ \\
        Load after epoch 3: $\bell_5 = (3, 1, 1)$ \\
        
        {\bf Mapped Hindsight Allocation} \\
        Load after epoch 1: $\bell_3^* = (1, 1, 1)$ \\
        Load after epoch 2: $\bell_4^* = (1, 2, 1)$ \\
        Load after epoch 3: $\bell_5^* = (2, 2, 1)$ \\
    \end{minipage}
    \vspace{-15pt}
    \caption{\small\em
        The graph on the left illustrates the nested instance $\wh{\calP}$ with $5$ rounds output by~\Cref{alg:mmo:nest} given instance $\calP$.
        On the right, we show the intermediate loads of \bricklay after each epoch and the loads of a hindsight solution we construct in \Cref{lem:mmo:nested_worst}.
        Note that this is not $\opt(\bl, \wh{\calP})$, but rather, demonstrates that $\opt(\bl, \calP)$ is still feasible (up to permutation) in $\wh{\calP}$. 
    }\label{fig:mmo:example2}
\end{figure}
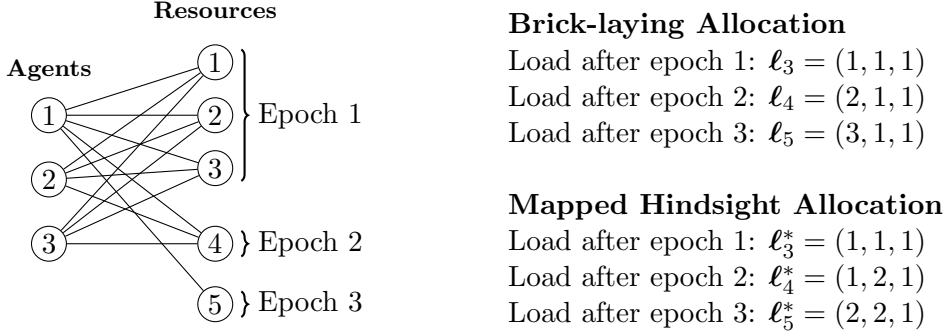

\newpage 

Against the strategy $\calP$ in \Cref{fig:mmo:example1}, \bricklay produces load vector $\bell = (3,1,1)$.
All three agents get at least one resource, and the first agent gets three resources; thus, the conjugate of the load vector is $\bell^{\circ} = (3, 1, 1, 0, 0)$.
Given $\calP$ as input, \Cref{alg:mmo:nest} constructs a nested strategy $\wh{\calP}$ with 3 non-trivial epochs, one for each non-zero element of the conjugate.
The first epoch has $\bell^{\circ}(1) = 3$ polymatroids with the bases $\chi_1,\chi_2,\chi_3$, since $\bell^{\circ}(0) = n = 3$.
The second epoch has $\bell^{\circ}(2) = 1$ polymatroid with again with bases $\chi_1,\chi_2,\chi_3$ because $\bell^{\circ}(1) = 3$.
Finally, the third epoch has $\bell^{\circ}(3) = 1$ polymatroid with a single basis $\chi_1$ since $\bell^{\circ}(2) = 1$.
\Cref{fig:mmo:example2} illustrates strategy $\wh{\calP}$.

The most technical hurdle of \Cref{lem:mmo:nested_worst} is proving that the hindsight optimal  $\opt(\bl, \calP)$ remains feasible (up to permutation) when the adversary plays strategy $\wh{\calP}$.
In our analysis, we leverage the brick representation of allocations.

\begin{definition}[Brick Representation]\label{def:back:multi_set}
    The brick representation of $\mbf{x}\in P$ on $\bell$ is $M_\mbf{x}^{\bell} = \{(i, l)\in [n]\times \Z_{\ge 0} \mid \bell(i) + 1 \le l \le \bell(i) + \mbf{x}(i)\}$.
    An element $(i,l)\in M_{\mbf{x}}^{\bell}$ is called a brick of $\mbf{x}$.
\end{definition}

Since the \bricklay allocations on $\calP$ and $\wh{\calP}$ are equivalent up to permutation, the union of their brick representations over all rounds is equivalent up to a relabeling of the agents.
This yields a natural bijection between \bricklay bricks on strategy $\calP$ with those bricks on strategy $\wh{\calP}$.
Our approach is to compose this natural bijection with a bijection from \bricklay bricks on $\calP$ with bricks in the optimal hindsight solution on $\calP$ to reconstruct the optimal hindsight allocation on $\wh{\calP}$ up to permutation.
This second bijection, described in \Cref{lem:mmo:maj_min_bij}, is carefully selected so that the reconstructed allocation remains feasible.
At a high level, if we map a hindsight brick $(i^*,l^*)$ to a later-arriving epoch $j$, we have to ensure it will have sufficiently high \bricklay load $\bell_h(i^*)$ to ensure feasibility; we require $i^*$ to have load at least $j$ so the mapped brick is feasible in $[\bell^{\circ}(j)]$.
\Cref{lem:mmo:maj_min_bij} guarantees that each brick $(i^*, l^*)$ has sufficiently large \bricklay load.

\begin{restatable}[Majorization Minimality Bijection]{lemma}{majminbij}\label{lem:mmo:maj_min_bij}
    Fix load vector $\bell\in\Z_{\ge 0}^n$, polymatroid $P$, and arbitrary basis $\mbf{x}^*\in \base(P)$.
    A basis $\mbf{x}\in\base(P)$ with $\mbf{x} + \bell \preceq \mbf{x}' + \bell$ for all bases $\mbf{x}'\in \base(P)$ admits a bijection $\phi: M_{\mbf{x}^*}^{\bell} \to M_{\mbf{x}}^{\bell}$ where $(i, l) = \phi(i^*, l^*)$ implies $\bell(i^*) + \mbf{x}(i^*) \ge l - 1$ for all bricks $(i^*, l^*)\in M_{\mbf{x}^*}^{\bell}$.
\end{restatable}

\Cref{lem:mmo:maj_min_bij} is easier to argue about in the special case when polymatroid $P$ is rank $1$.
In this case, every basis has a single brick in its brick representation.
A majorization minimal basis will allocate a brick $(i, \bell(i) + 1)$ to the agent $i$ with the lowest load.
Since $\bell(i)$ had minimal load, every other $i^*$ satisfies $\mbf{x}(i^*) + \bell(i^*) \ge (\bell(i) + 1) - 1$; any other singleton brick will admit the desired bijection.
The proof of \Cref{lem:mmo:maj_min_bij} in its full generality necessitates polymatroid machinery, so we defer it to \Cref{app:sec:bl_mmo}.
These ideas may be of independent interest.

\section{Majorization Minimax-Optimality: Formal Proofs}\label{sec:mmo}
With the intuition of \Cref{sec:tech_overview} in mind, we prove \Cref{lem:mmo:nested_worst}.
In our proofs, the $i^{\text{th}}$ component of the vectors $\mbf{x}^\uparrow$ and $\mbf{x}^\downarrow$ are the $i^{\text{th}}$ smallest and $i^{\text{th}}$ largest of a vector $\mbf{x}$, respectively.

\nested*
\begin{proof}
    Define $(\mbf{x}_t)_{t\in [h]}$ to be the allocations produced by \bricklay on strategy $\calP$, which admit load vectors $\bell_{t} \define \sum_{s\in[t]}\mbf{x}_s$ and $\bell_0 = \vec{0}$.
    We also take $\bell\define\bell_h = \bell(\bl, \calP)$ for convenience.
    The conjugate $\bell^{\circ}(j)$ is non-increasing in $j$, thus the strategy $\wh{\calP}$ is nested.
    
    We establish some notation.
    Let $\wh{P}_1,\dots, \wh{P}_m$ be the polymatroids that are deterministically played by $\wh{\calP}$.
    Notice the sequence has length $m$, because each polymatroid is associated with one resource.
    Define $\tau_j = \sum_{k\in[j]} \bell^{\circ}(k)$ to be the last round in epoch $j$.
    Take $(\wh{\mbf{x}}_\tau)_{\tau\in[m]}$ and $(\wh{\bell}_\tau)_{\tau\in[m]\cup\{0\}}$ to be the allocations and load vectors generated by \bricklay on $\wh{\calP}$.
    Let $(\mbf{x}_t^*)_{t\in[h]}$ be a feasible allocation for the polymatroids realized by $(\bl, \calP)$ with $\opt(\bl, \calP) = \sum_{t\in [h]} \mbf{x}_t^*$;
    this hindsight solution is deterministic since $\bl$ and $\calP$ are deterministic.
    Finally, $M_t \define M_{\mbf{x}_t}^{\bell_{t-1}}$ and $M_t^* \define M_{\mbf{x}_t^*}^{\bell_{t-1}}$ are the brick representations of $(\mbf{x}_t)_{t\in [h]}$ and $(\mbf{x}_t^*)_{t\in [h]}$ on round $t$, respectively. 

    We show that $\bell(\bl, \calP) \sim \bell(\bl, \wh{\calP})$.
    For induction on epochs $j$, assume that $\wh{\bell}_{\tau_j}(i) = j$ for all $i\in [\bell^{\circ}(j)]$.
    The claim holds for $j=0$ since $\tau_0 = 0$ and $\wh{\bell}_0 = \vec{0}$.
    Consider epoch $j+1$.
    In this epoch, the adversary presents $\bell^{\circ}(j+1)$ copies of polymatroid $\wh{P}$ with rank function $\wh{r}(A) = \1{A\cap [\bell^{\circ}(j)] \neq \emptyset}$; the feasible bases are $\{\chi_i\mid i\in [\bell^{\circ}(j)]\} = \base(\wh{P})$.
    Since all agents in $[\bell^{\circ}(j)]$ have the same level at the beginning of the epoch and \bricklay allocates to the lowest load agent, breaking ties in favor of lower indices, \bricklay will select bases $\chi_1,\dots,\chi_{\bell^{\circ}(j+1)}$ in epoch $j+1$.
    For $i\in [\bell^{\circ}(j+1)]$ this gives $\wh{\bell}_{\tau_{j+1}}(i) = \wh{\bell}_{\tau_{j}}(i) + 1 = j+1$.
    The inductive claim implies that $\wh{\bell}_m(i) \ge \wh{\bell}_{\tau_j}(i) \ge j$ where $j$ is the latest epoch with $i\in [\bell^{\circ}(j)]$.
    This inequality, in addition to conjugation giving $j = \bell^{\downarrow}(i)$ and yields $\wh{\bell}_m(i) \ge \bell^{\downarrow}(i)$.
    By definition $\wh{\bell}_m([n]) = \bell([n]) = m$; it must be that $\wh{\bell}_m(i) = \bell^{\downarrow}(i)$ and further $\bell(\bl, \calP) = \bell \sim \wh{\bell}_m(i) = \bell(\bl,\wh{\calP})$.

    We now show $\opt(\bl, \calP)\succeq \opt(\bl, \wh{\calP})$, by mapping bricks in $(\mbf{x}_t^*)_{t\in [h]}$ to a new allocation $(\wh{\mbf{x}}_\tau^*)_{\tau\in [m]}$ which is feasible on $\wh{\calP}$ and yields a load vector $\wh{\bell}^*$ sharing an equivalence class with that of $\sum_{t\in[h]}\mbf{x}_t^* = \opt(\bl, \calP)$.
    This implies the claim, as $\opt(\bl, \wh{\calP}) \preceq \hat{\bell}^*\sim \opt(\bl, \calP)$ by the majorization-minimality of optimal hindsight solutions.
    Each hindsight allocation brick $(i^*,l^*)$ has an associated agent $i^*$ and level $l^*$, which we map independently.
    First, we map agents according to the bijection $\kappa:[n]\to[n]$ for which $\bell(i) < \bell(i')$ implies $\kappa(i) > \kappa(i')$ for all $i,i'$.
    To map levels to a round, we compose several functions.
    First, for each $t\in[h]$ take $\phi_t: M_t^* \to M_t$ to be a bijection with $(i, l) = \phi_t(i^*, l^*)$ implying $\bell_{t-1}(i^*) + \mbf{x}_t(i^*) \ge l - 1$ for all $(i^*,l^*)\in M_t^*$, which is guaranteed to exist by \Cref{lem:mmo:maj_min_bij}.
    The last function is a bijection $\delta: \bigcup_{t\in[h]} M_t \to [m]$ that maps bricks of $(\mbf{x}_t)_{t\in [h]}$ to a round on strategy $\wh{\calP}$.
    We choose $\delta$ so that $\delta(i, l) = \tau$, where $\tau$ is in epoch $j\define l$.
    Such a bijection exists since epoch $j$ has $\bell^{\circ}(j)$ rounds, which is exactly the number of bricks in $M_t$ with level $l$.
    Now we can construct $(\wh{x}_\tau^*)_{\tau\in[m]}$.
    For each $t\in [h]$ and $(i^*,l^*)\in M_t^*$, we select allocation $\wh{\mbf{x}}_{\tau}^* \gets \chi_{\kappa(i^*)}$ on round $\tau = \delta\circ \phi_t(i^*,l^*)$.

    It remains to show that $(\wh{\mbf{x}}_\tau^*)_{\tau\in [m]}$ is feasible on $\wh{\calP}$ and $\sum_{t\in [h]} \mbf{x}_\tau^* \sim \sum_{\tau\in[m]} \wh{\mbf{x}}_t^*$.
    Since we relabeled agents according to $\kappa$, the loads of agents will be equal up to permutation.
    This gives $\sum_{t\in [h]} \mbf{x}_\tau^* \sim \sum_{\tau\in[m]} \wh{\mbf{x}}_t^*$.
    We now consider feasibility.
    We claim the composition $\delta \circ \phi_t$ maps each input to a different round.
    This holds because each bijection $\phi_t$ has codomain $M_t$, which are disjoint for all $t$, and $\delta$ is a bijection. 
    Lastly, we argue $\wh{\mbf{x}}_\tau^* \in \wh{P}_\tau$ for all $\tau\in[m]$.
    For any $t\in[h]$ and $(i^*, l^*)\in M_t^*$ define $(i,l) \define \phi_t(i^*, l^*)$.
    Since brick $(i^*, l^*)$ will appear in epoch $l$, we require that $\kappa(i^*) \in [\bell^{\circ}(l-1)]$ for our allocation to be a base of the polymatroid $P$ presented in epoch $l$.
    \Cref{lem:mmo:maj_min_bij} gives $\bell(i^*) \ge \bell_{t-1}(i^*) + \mbf{x}_t(i^*) \ge l - 1$, which implies that $i^*$ has one of the $\bell^{\circ}(l-1)$ largest loads among agents.
    Our bijection on agents gives $\kappa(i^*)\in [\bell^{\circ}(l - 1)]$, which completes the proof of feasibility.
\end{proof}

We now present the second half of our supporting claims to prove \Cref{thm:intro:mmo}.
The goal is to show that \bricklay performs weakly better than any alternative allocation strategy against nested strategies.
This involves relabeling of agents so that alternative strategies that do not maintain relatively equal loads on agents cannot allocate to low load agents in the future.
We call the adversarial strategy that performs this careful relabeling a nested response strategy. 

\begin{definition}[Nested Response Strategy]
    Let $z_1,\dots,z_m\in\Z_{\ge 0}$ be a non-increasing sequence of integers.
    Given a history of play prior to round $t$, the nested response strategy $\calP$ for allocation strategy $\calA$ seeded by $(z_t)_{t\in [m]}$ presents the polymatroid $P_t$ on round $t$ with rank $r_t(A) = \1{A\cap N_t \neq \emptyset}$, where $N_t = \arg\max\{\bell_{t-1}(A) \mid |A| = z_t\}$ for intermediate loads $\bell_{t-1}$ produced by $\calA$.
\end{definition}

\begin{observation}[Nested Response Strategy]
    The nested response strategy $\calP$ for any allocation strategy $\calA$ seeded by non-increasing $(z_t)_t$ is a nested strategy.
\end{observation}

\begin{proof}
    Let $P_1,\dots, P_m$ be the polymatroids realized by $\calP$ against any allocation strategy.
    Loads produced by the allocation strategy are non-decreasing each round and $(z_t)_t$ is non-increasing, so by construction $i\in N_t$ implies $i\in N_{t-1}$.
    The strategy $\calP$ is nested.
\end{proof}

We now show that the \bricklay algorithm is the most resistant algorithm to nested response strategies for any seed sequence.

\begin{restatable}[\bricklay vs. Alternatives]{lemma}{vsalt}\label{lem:mmo:vs_alt}
    Fix a non-increasing integer sequence $(z_t)_t$ and strategy $\calA'$.
    Let $\wh{\calP}$ and $\calP'$ be the nested response strategies seeded with $(z_t)_t$ on $\bl$ and $\calA'$, respectively.
    These satisfy $\bell(\bl, \wh{\calP}) \preceq \bell(\calA', \calP')$ and $\opt(\bl, \wh{\calP}) \sim \opt(\calA', \calP')$ with probability $1$.
\end{restatable}

\begin{proof}
    We first focus on the property $\opt(\bl, \wh{\calP}) \sim \opt(\calA', \calP')$.
    Let $P_1',\dots,P_m'$ be the polymatroids realized by $(\calA', \calP')$.
    All realizations of the sequence $(P_t')_{t\in [m]}$ admit the same feasible allocations, up to permutation of the agents.
    In addition, these realizations further share allocations, up to permutation, with the sequence $(\wh{P}_t)_t$.
    This implies that the majorization minimal element of every realization $(\calP_t')_t$ is equivalent up to permutation to the majorization minimal element of $(\wh{P}_t)_t$.
    This immediate yields $\opt(\bl, \wh{\calP}) \sim \opt(\calA', \calP')$ with probability $1$.

    We now show $\bell(\bl, \wh{\calP}) \preceq \bell(\calA', \calP')$ with probability $1$.
    Fix an arbitrary realization of polymatroids $(\wh{\calP}_t)_t$, allocations $(\wh{\mbf{x}}_t)_t$, and loads $(\wh{\bell}_t)_t$ produced by profile strategy $(\bl, \wh{\calP})$.
    Likewise, let $(\calP_t')_t$, $(\mbf{x}_t')_t$, and $(\bell_t')_t$ be an arbitrary realization of the same on profile $(\calA', \calP')$.
    Without loss of generality, we relabel agents so that $\wh{\bell}_m(1)\ge\dots\ge\wh{\bell}_m(n)$ and $\bell_m'(1)\ge\dots\ge\bell_m'(n)$; this also implies that $\wh{N}_t = N_t'$ for each $t$.
    For induction on $k$ from $n$ to $0$, we assume that $\wh{\bell}_m([n]\setminus [k]) \ge \bell_m'([n]\setminus [k])$.
    The claim holds for $k = n$ by definition.
    Consider the claim for $k-1$.
    Let $\tau = \max\{t \mid k-1 \in N_t\}$ be the latest round on which there is a feasible allocation to $k-1$.
    The total number of allocated resources up to including this round is $\tau$.
    Further, the total load on agents in $[k]$ satisfy: $\wh{\bell}_{\tau}([k]) = \tau - \wh{\bell}_{\tau}([n]\setminus [k]) = \tau - \wh{\bell}_m([n]\setminus [k])$, where the last equality holds because agents in $[n]\setminus [k]$ will no longer receive allocations after round $\tau$.
    The same equation holds with $\wh{\bell}$ replaced by $\bell'$.
    \bricklay will allocate resources on $\wh{\calP}$ as evenly as possible; loads will differ by at most $1$ on the agents in $[n]\setminus [k]$.
    Further, $N_\tau$ is selected to contain the agents with low load.
    This gives: $\wh{\bell}_{\tau}(k) = \floor{\frac{\wh{\bell}_{\tau}([k])}{k}}$
    On the other hand, $\calA'$ makes no such guarantees on the evenness of loads $\bell'$ on agents in $[k]$.
    This in conjunction with the fact that $N_{\tau}$ is selected to contain agents with smallest load, yields $\bell_\tau'(k) \le \floor{\frac{\bell_{\tau}'([k])}{k}}$.
    Combining all these facts proves the inductive claim:
    {\allowdisplaybreaks
    \begin{align*}
        \wh{\bell}_m([n]\setminus [k-1])
        &= \wh{\bell}_\tau([n]\setminus [k]) + \wh{\bell}_\tau(k)
        &&\text{linearity and }\tau\text{ definition} \\
        &= \wh{\bell}_\tau([n]\setminus [k]) + \floor{\frac{\wh{\bell}_{\tau}([k])}{k}}
        &&\bl\text{ load equity} \\
        &= \wh{\bell}_\tau([n]\setminus [k]) + \floor{\frac{\tau - \wh{\bell}_m([n]\setminus [k])}{k}}
        &&\text{Argued fact} \\
        &= \floor{\frac{\tau + (k-1)\cdot \wh{\bell}_{m}([n]\setminus [k])}{k}}
        &&a + \floor{b} = \floor{a+b} \text{ for }a\in\Z_{\ge 0} \\
        &\ge \floor{\frac{\tau + (k-1)\cdot \bell_{m}'([n]\setminus [k])}{k}}
        &&\text{Inductive assumption} \\
        &= \bell'_{\tau}([n]\setminus [k]) + \floor{\frac{\tau - \bell_m'([n]\setminus [k])}{k}} \\
        &= \bell_\tau'([n]\setminus [k]) + \floor{\frac{\bell'_{\tau}([k])}{k}} \\
        &\ge \bell_{\tau}'([n]\setminus [k]) + \bell_{\tau}'([n]\setminus [k]) \\
        &= \bell_{m}'([n]\setminus [k-1])
        \tag*{\qedhere}
    \end{align*}}
    The inductive assumption is exactly the condition for majorization specified in \Cref{def:intro:maj} by the assumption that $\wh{\bell}$ and $\bell'$ are sorted, so we have $\wh{\bell}_m \preceq \bell'_m$.
\end{proof}

\subsection{Majorization Minimax Optimality}
With \Cref{lem:mmo:nested_worst,lem:mmo:vs_alt}, we can conclude that \bricklay is \mmo.
\blmmo*
\begin{proof}
    The claim is the result of composing \Cref{lem:mmo:nested_worst} and \Cref{lem:mmo:vs_alt}.
    Fix an allocation strategy $\calA'$ and an arbitrary adversarial strategy $\calP$.
    Define nested strategy $\wh{\calP}$ to be the output of \Cref{alg:mmo:nest} on $\calP$.
    The nested strategy $\wh{\calP}$ has a chain of neighborhoods $N_1\supseteq\dots\supseteq N_m$.
    Take $\calP'$ to be the nested response strategy on $\calA'$ seeded by $z_t \define |N_t|$.
    Notice that $\wh{\calP}$ is the nested response strategy on $\bl$ seeded by $(z_t)_t$, so the guarantees of \Cref{lem:mmo:vs_alt} hold.  
    \Cref{lem:mmo:nested_worst,lem:mmo:vs_alt} and transitivity of majorization yield the claim:
    \begin{align*}
        &\bell(\bl, \calP) \sim \bell(\bl, \wh{\calP}) \preceq \bell(\calA', \calP') \\
        &\opt(\bl, \calP) \succeq \opt(\bl, \wh{\calP}) \sim \opt(\calA', \calP')
        \qedhere
    \end{align*}
\end{proof}

% %As informally illustrated in \Cref{subsec:intro:mmo}, \mmo implies minimax optimal regret.
% We formalize these claims for the \bricklay strategy.
% Also, as a consequence of \Cref{cor:intro:bl_opt_regret}, \bricklay also achieves minimax optimal competitive ratios.
\bloptregret*
\begin{proof}
    We prove the corollary for Schur-concave $f$. Similar logic shows the claim holds for Schur-convex $g$.
    For an alternative strategy $\calA'$ and adversarial strategy $\calP\in\calP_{n,m}$, \Cref{thm:intro:mmo} guarantees there is another strategy $\calP'\in{n,m}$ satisfying:
    \begin{align*}
        c_{f,\alpha}(\bl, \calP)
        &= \alpha \cdot f(\opt(\bl, \calP)) - f(\bell(\bl, \calP)) \\
        &\le \alpha\cdot f(\opt(\calA', \calP')) - f(\bell(\calA', \calP'))
        &\text{\Cref{thm:intro:mmo} and Schur-concavity} \\
        &= c_{f,\alpha}(\calA', \calP')
    \end{align*}
    Since all $\calP$ admit $\calP'$ satisfying the above, we have $\reg_{f,\alpha}^{n,m}(\bl) \le \reg_{f,\alpha}^{n,m}(\calA')$
\end{proof}
\newpage

\bibliographystyle{alpha}
\bibliography{references}

\appendix

\section{Majorization and Equity}\label{app:sec:maj}
Majorization is a fundamental mathematical tool to analyze equity.
In addition to its canonical definition, majorization has many natural equivalent conditions.
We use $\pp{\cdot} = \max(0, \cdot)$ to denote the positive parts function.

\begin{fact}[Majorization Equivalences]
    Majorization $\mbf{x} \preceq \mbf{y}$ for $\mbf{x},\mbf{y}\in \Z_{\ge 0}^n$ is equivalent to each of the following:
    \begin{itemize}
        % \item Increasing partial sums: $\mbf{x}^{\downarrow}([k]) \ge \mbf{y}^{\downarrow}([k]), \forall k\in[n]$
        \item Decreasing partial sums: $\mbf{x}^{\uparrow}([k]) \le \mbf{y}^{\uparrow}([k]), \forall k\in [n]$
        \item Thresholding inequality\footnote{Assuming $\mbf{x}([n]) = \mbf{y}([n])$, the thresholding inequality is equivalent to Karamata's inequality, which requires $\sum_{i\in[n]} f(\mbf{x}(i)) \ge \sum_{i\in[n]} f(\mbf{y}(i))$ for all convex $f:\Z_{\ge 0}^n \to \R$.}: $\sum_{i\in[n]}\pp{\mbf{x}(i) - \gamma} \ge \sum_{i\in [n]}\pp{\mbf{y}(i) - \gamma}, \forall \gamma \in \Z_{\ge 0}$
        \item Hardy, Littlewood, and Polya \cite{hardy1988inequalities}: $\mbf{y} = W\mbf{x}$ for some doubly stochastic $W\in\R^{n\times n}_{\ge 0}$.
    \end{itemize}
\end{fact}

There is a plethora of Schur-concave and Schur-convex functions.
\cite{marshall2011majorization} contains an extensive list of functions in these classes.
Below, we provide a sample of these functions, which take as input a load vector $\bell\in\Z_{\ge 0}^n$, use $\mu = \frac{1}{n}\sum_{i\in[n]} \bell(i)$, and may have parameter $b\ge 0$:

\medskip
\noindent
\begin{minipage}[t]{0.48\textwidth}
    \textbf{Schur-concave Functions}
    \begin{itemize}[leftmargin=*]
        \item $b$-Matching: $\sum_{i\in[n]}\min(b, \bell(i))$
        \item $b$-Subsidized NSW: $\prod_{i\in [n]}(\bell(i) + b)^{1/n}$
        \item Egalitarian Welfare: $\min_{i\in[n]}{\bell(i)}$
        \item Power Means: $\left(\sum_{i\in[n]}|\bell(i)|^p\right)^{1/p}$ for $p < 1$
    \end{itemize}
\end{minipage}
\hfill
\begin{minipage}[t]{0.48\textwidth}
    \textbf{Schur-convex Functions}
    \begin{itemize}[leftmargin=*]
        \item Gini Index: $\frac{1}{\left(2n^2 \mu\right)}\left( \sum_{i,j\in[n]} |\bell(i) - \bell(j)| \right)$
        \item Latency: $\frac{1}{2}\sum_{i\in[n]}\left(\bell(i) + 1\right)\bell(i)$
        \item Makespan: $\max_{i\in[n]}{\bell[i]}$
        \item $\ell^p$-norms: $\left(\sum_{i\in[n]}|\bell(i)|^p\right)^{1/p}$ for $p \ge 1$
    \end{itemize}
\end{minipage}

\section{Combinatorial Constraints}\label{app:sec:combo_constraint}
\subsection{Polymatroids}
Polymatroids generalize matroids to multisets \cite{herzog2002discrete,fujishige2005submodular}.
Like matroids, polymatroids possess several axiomatic properties that make them conducive to greedy optimization.
Given vectors $\mbf{x}, \mbf{y}$, we use $\mbf{x} \le \mbf{y}$ when $\mbf{x}(i) \le \mbf{y}(i),\forall i$.

\begin{restatable}[Polymatroid Axioms \cite{herzog2002discrete}]{fact}{polyax}\label{fact:back:poly_ax}
    A discrete polymatroid $P$ satisfies each of the following:
    \begin{itemize}
        \item Zero-vector membership: $\vec{0} \in P$
        \item Downward closure: if $\mbf{x} \in P$ and $\mbf{y} \le \mbf{x}$ for $\mbf{y}\in\Z_{\ge 0}^n$ then $\mbf{y}\in P$
        \item Exchange property: if $\mbf{x}, \mbf{y}\in P$ with $\mbf{x}([n]) < \mbf{y}([n])$, then there is $i\in[n]$ satisfying $\mbf{y}(i) > \mbf{x}(i)$ and $\mbf{x} + \chi_i \in P$.
    \end{itemize}
\end{restatable}

\subsection{Convex Games}
A convex game is a special case of a cooperative game among $n$ agents.
These games admit a non-empty core -- outcomes that cannot be improved upon by any coalition of players -- that can be characterized by a simple greedy algorithm \cite{shapley1971cores}.
Discrete convex games are characterized by a characteristic function $v:2^{[n]}\to \Z_{\ge 0}$ that maps a coalition of players to the value they can collectively generate.
The function must be normalized $v(\emptyset) = 0$ and supermodular $v(A) + v(B) \le v(A\cup B) + v(A\cap B)$ for all $A,B\subseteq [n]$.
The core of the game is the set of payoff vectors $\mbf{x}\in \Z_{\ge 0}$ that are efficient $\mbf{x}([n]) = v([n])$ and incentivize every coalition to participate in the grand coalition $\mbf{x}(A) \ge v(A)$ for all coalitions $A\subseteq [n]$.

Our results extend to the problem of computing equitable cores in convex games, as they are intimately related to bases of polymatroids.
Specifically, a vector is a basis of the polymatroid with rank function $r([n] \setminus A) = v([n]) - v(A)$ if and only if the same vector is a core of the convex game with characteristic function $v$ \cite{fujishige2005submodular}.

\section{Proof of Brick-Laying Properties}\label{app:sec:bl}
\blmajmin*

\begin{proof}
    This proof has three components.
    We first argue efficiency.
    We have value oracle access to $r_t$ by assumption and can thus implement a value oracle for $\wt{r}_t = r_t + \bell_{t-1}$ in polynomial time.
    A value oracle of $\wt{r}_t$, defining polymatroid $\wt{P}_t$, can be used to efficiently implement an independence oracle for $\wt{P}_t$ that decides $\mbf{x}\in\wt{P}_t$.
    We have $\mbf{x}\in\wt{P}_t$ if and only if $\mbf{x} \ge 0$ and $\min_{A} (\wt{r}_t(A) - \mbf{x}(A)) \ge 0$.
    The optimization problem can be solved in polynomial time with value oracle access to $\wt{r}_t$ \cite{schrijver2000submodular,iwata2001combinatorial}.
    We can efficiently identify  $i_k$ by enumerating agents and applying the independence oracle.
    All other operations require polynomial time, thus \Cref{alg:back:bl} is efficient.

    Let $\wt{\mbf{x}}$ be the value of $\mbf{x}^{(k)}$ after the while loop in \Cref{alg:back:bl} terminates.
    First, we show that a basis $\wt{\mbf{x}}\in\wt{\calP}_t$ takes the form $\mbf{x}_t + \bell_{t-1}$ for some basis $\mbf{x}_t \in \calP_t$.
    If $\wt{\mbf{x}} \ge \bell_{t-1}$, then $\wt{\mbf{x}} - \bell_{t-1} \in \base(P_t)$ and the claim is true.
    Now, for contradiction, assume that $\wt{\mbf{x}}(i) < \bell_{t-1}(i)$ on some $i\in [n]$.
    Then we have
    \begin{align*}
        \wt{\mbf{x}}([n])
        &= \wt{\mbf{x}}([n] - i) + \wt{\mbf{x}}(i) \\
        &< \wt{r}_t([n] - i) + \bell_{t-1}(i)
        &\text{Polymatroid constraint and contradictory assumption} \\
        &\le r_t([n]) + \bell_{t-1}([n])
        &\text{Monotonicity and }\wt{r}_t = r_t + \bell_{t-1} \\
        &= \wt{r}_t([n])
    \end{align*}
    A contradiction arises as $\wt{\mbf{x}}$ is not a basis of $\wt{\calP}$.
    This completes the claim.

    Lastly, we prove $\mbf{x}_t + \bell_{t-1}$ is majorization minimal.
    It suffices to show $\wt{\mbf{x}}$ is a majorization minimal basis of $\wt{\calP}$.
    \cite{federgruen1986greedy} provides a greedy algorithm to find a polymatroid basis that maximizes functions of the form $v(\mbf{x}) = \sum_{i\in [n]} v'(\mbf{x}(i))$ for concave and non-negative $v'$.
    If we can simultaneously maximize objectives of the form $v'(\cdot) = \min(\gamma, \cdot)$ for different thresholds $\gamma\in\Z_{\ge 0}$, we would find a majorization minimal basis (maximization of all such $v'$ is equivalent to threshold characterization of majorization).
    \Cref{alg:back:bl} is a special case of the algorithm by \cite{federgruen1986greedy} that is consistent with maximizing $v'$ at every threshold $\gamma$ simultaneously, thus it finds a majorization minimal basis.
\end{proof}

% \begin{proof}
%     Let $\wt{\mbf{x}}\in\wt{P}$ be an arbitrary basis.
%     If $\wt{\mbf{x}} \ge \bell$, then $\wt{\mbf{x}} - \bell \in \base(P)$ and the claim is true.
%     Now, for contradiction, assume that $\wt{\mbf{x}}(i) < \bell(i)$ on some $i\in [n]$.
%     Then we have
%     \begin{align*}
%         \wt{\mbf{x}}([n])
%         &= \wt{\mbf{x}}([n] - i) + \wt{\mbf{x}}(i) \\
%         &< \left(r([n] - i) + \bell([n] - i)\right) + \bell(i)
%         &\text{Polymatroid constraint} \\
%         &\le r([n]) + \bell([n])
%         &\text{Monotonicity} \\
%         &= \wt{r}([n])
%     \end{align*}
%     A contradiction arises as $\wt{\mbf{x}}$ is not a base.
%     This completes the claim.
% \end{proof}

\hindalloc*
\begin{proof}
    Let $r_1,\dots,r_h$ be the rank functions of the realized polymatroids $P_1,\dots, P_h$.
    Define polymatroid $Q$ using the rank function $\sum_{t\in[h]} r_h$, which is valid since submodular functions are closed under addition \cite{schrijver2002combinatorial,fujishige2005submodular}.
    Notice that $H = \base(Q)$, thus \Cref{alg:back:bl} computes a majorization minimal basis $\bell^*\in H$.
    The hindsight load vector is $\opt(\calA, \calP) \define \bell^*$, which is randomized according to the realization of $P_1,\dots, P_h$.
\end{proof}

\section{Brick-Laying is Majorization Minimal-Optimal Deferred Proofs}\label{app:sec:bl_mmo}
In this section, we develop polymatroid machinery to support our arguments in \Cref{lem:mmo:nested_worst}.
The final goal of this section is to show \Cref{lem:mmo:maj_min_bij}, but we require a variant of an exchange argument before we can complete our claims.
We use auxiliary matroids to prove this exchange property.

\begin{lemma}[Auxiliary Matroid]
    Fix polymatroid $P\subseteq \Z_{\ge 0}^n$ with rank $r\define \rank(P)$ and define surjection $\zeta: 2^{[n]\times[r]} \to \Z_{\ge 0}^n$ which gives $\mbf{x} \define \zeta(S)$ with $\mbf{x}(i) = \left|\{(i',j')\in S\mid i' = i\}\right|$.
    The set system $\calM_P = ([n]\times [r], \calI)$ where $S\in\calI$ is independent if and only if $\zeta(S)\in P$ is a matroid, is called the auxiliary matroid of $P$.
    Further, a set is a base $A\in \base(\calM_P)$ if and only if $\zeta(A) \in \base(P)$.
\end{lemma}

\begin{proof}
    The empty-set is independent in matroid $\calM_P$ as $\zeta(\emptyset) = \vec{0} \in P$.
    To show downward closure, take $A \subseteq B$ with independent $B\in\calI$.
    By construction $\zeta(A) \le \zeta(B)$ with $\zeta(B)\in P$;
    these facts imply $\zeta(A)\in P$ which gives downward closure: $A\in\calI$.

    We take independent sets $A,B \in\calI$ with $|A| < |B|$ to show the exchange property.
    Define $\mbf{x}_A = \zeta(A)$ and $\mbf{x}_B = \zeta(B)$, which satisfy $\mbf{x}_A([n]) < \mbf{x}_B([n])$.
    By the exchange property of polymatroids, there is an $i$ with $\mbf{x}_A(i) < \mbf{x}_B(i)$ for which $\mbf{x}_A + \chi_i \in P$.
    With the definition $T_i = \{(i', j')\mid i' = i\}$, notice that $T_i\cap (B\setminus A)$ is non-empty.
    Select any $(i,j)\in T_i\cap(B\setminus A)$ which will have $\mbf{x}_A + \chi_i = \zeta(A + (i,j))$.
    This shows that we can exchange $(i,j)\in B\setminus A$ into an independent set $A$ as desired.

    Finally, since $|A| = \mbf{x}_A([n])$ for $\mbf{x}_A \define \zeta(A)$ we have $|A| > r$ implies $A\notin \calM_P$ and $|A| = r$ for any $A$ in the $\zeta$ preimage of some $\mbf{x}\in P$.
    Thus, $A\in \base(\calM_P)$ if and only if $\zeta(A)\in\base(P)$.
\end{proof}

\begin{restatable}[Inclusion/Exclusion Exchange]{lemma}{exchange}\label{lem:mmo:exchange}
    Let $P$ be a polymatroid with bases $\mbf{x},\mbf{x}^*\in P$.
    If $\mbf{x}(A) > \mbf{x}^*(B)$ for $A\subseteq B$, then there is $i\in A$ and $i^*\in B^c$ such that $\mbf{x} - \chi_i + \chi_{i^*} \in\base(P)$.
\end{restatable}

\begin{proof}
    We first argue that $A$ and $B^c$ are non-empty.
    The set $A$ is nonempty since $\mbf{x}(A) > \mbf{x}^*(B) \ge 0$.
    Further $B^c$ is nonempty since $\mbf{x}(B^c) = \mbf{x}([n]) - \mbf{x}(B) > \mbf{x}([n]) - \mbf{x}(A) \ge 0$.

    We use several definitions to facilitate our proof.
    Let $r\define \rank(P)$ be the rank of $P$.
    Take $E_A = \{(i,j)\in[n]\times [r]\mid i\in A\}$ and define $E_B$ in a similar manner for $B$.
    Let $S$ and $S^*$ be such that $\zeta(S) = \mbf{x}$ and $\zeta(S) = \mbf{x}^*$.
    Notice that $|S\cap E_A| = \mbf{x}(A) > \mbf{x}^*(B) = |S^*\cap E_B|$.
    We now apply a variant of the exchange property specified in \cite[Corollary 39.12a]{schrijver2002combinatorial}.
    This variant of exchange guarantees that there is a bijection $\psi:(S\setminus S^*)\to (S^*\setminus S)$ such that $S - e + \psi(e)\in \calM_P$ for all $e\in S\setminus S^*$.
    We claim that the image of $(S\setminus S^*)\cap E_A$ under this bijection satisfies $\psi((S\setminus S^*)\cap E_A) \supsetneq (S^*\setminus S)\cap E_B$.
    Since images of a bijection preserve cardinality $|\psi(S)| = |S|$ for all $S$, the claim is a result of:
    \begin{align*}
        &|(S^*\setminus S)\cap E_B| \\
        &= |(S^*\setminus S)\cap E_B| + |S\cap S^*\cap E_B| - |S\cap S^*\cap E_B| \\
        &= |(S^* \cap E_B)| - |S\cap S^*\cap E_B| \\
        &< |(S^* \cap E_A)| - |S\cap S^*\cap E_A|
        &&|(S^* \cap E_A)| < |(S^* \cap E_A)| \text{ and }E_A \subseteq E_B \\
        &= |(S\setminus S^*)\cap E_A|
    \end{align*}
    The bijection guarantees the existence of $(i, j)\in (S\setminus S*)\cap E_A$ and $(i^*, j^*) \in (S^*\setminus S)\cap E_B^c$ such that $S - (i,j) + (i^*, j^*)\in\calM_P$.
    The auxiliary matroid gives $\zeta(S - (i,j) + (i^*, j^*)) = \mbf{x} - \chi_i + \chi_{i^*} \in \base(P)$ where $i\in A$ and $i^*\in B^c$ as claimed.
\end{proof}

It remains to show \Cref{lem:mmo:maj_min_bij}.
Ideas from a conversation with Claude (Sonnet 4.6) were used to generate this proof, most notably of which is the use of Hall's theorem.
See \Cref{app:sec:ai} for additional details.

\majminbij*

\begin{proof}
    To prove the claim, we assume no such bijection exists and use Hall's theorem to identify a base $\wt{\mbf{x}}$ creating the contradiction $\mbf{x} + \bell \npreceq \wt{\mbf{x}} + \bell$.    

    Define $L_{k} = \{i\mid \bell(i) + \mbf{x}(i) \ge k\}$ to be the set of agents with large load under allocation $\mbf{x}$; notice that $L_{k} \subseteq L_{k-1}$.
    We now define a bipartite graph.
    The two classes of nodes are $M_{\mbf{x}}^{\bell}$ and $M_{\mbf{x}^*}^{\bell}$.
    Node $(i, l)\in M_{\mbf{x}}^{\bell}$ has an edge to $(i^*, l^*) \in M_{\mbf{x}}^{\bell}$ if and only if $i^*\in L_{l-1}$.
    Notice that a perfect matching on this bipartite graph defines a bijection $\phi$ with the desired properties.
    Since we assume no such bijection exists. Hall's theorem implies the existence of $A\subseteq M_{\mbf{x}}^{\bell}$ with $|A| > |N(A)|$, where $N(A)$ denotes the neighborhood of $A$ on the bipartite graph.
    For an arbitrary $A$, we first compute $|N(A)|$.
    Let $k = \min\{l \mid (i, l)\in A\}$ and notice that there is $(i,k)\in A$ with an edge to all $(i^*, \ell^*)\in M_{\mbf{x}^*}^{\bell}$ with $i^*\in B_{k-1}$.
    Nodes $(\wh{i}, l)\in A$ with $l \ge k$ have a weakly smaller neighborhood.
    Simplifying shows:
    \begin{align*}
        |N(A)| = \left|\left\{ (i^*,l^*)\in  M_{\mbf{x}^*}^{\bell} \mid i^*\in L_{k-1} \right\}\right| = \mbf{x}^*(L_{k-1})
    \end{align*}
    Without loss of generality, we take $A = \{(i,l)\in M_{\mbf{x}}^{\bell} \mid l \ge k \}$, as it does not increase alter $N(A)$.
    We now compute the cardinality of such a set:
    \begin{align*}
        |A|
        &= \left|\left\{(i,l)\in M_{\mbf{x}}^{\bell} \mid l \ge k\right\}\right| \\
        &= \left|\left\{(i,l) \mid \bell(i) + 1 \le l \le \bell(i) + x(i) \text{ and } k< \bell(i) + 1\right\}\right| \\
        &\quad\quad\quad + \left|\left\{(i,l) \mid k \le l \le \bell(i) + x(i) \text{ and } k \ge \bell(i) + 1\right\}\right| \\
        &= \sum_{i\in L_{k}: k\le\bell(i) + 1} \mbf{x}(i) + \sum_{i\in L_{k}: k> \bell(i) + 1} \left(\mbf{x}(i) + \bell(i) - k + 1)\right) \\
        &\le \mbf{x}(L_k)
    \end{align*}
    In summary, by Hall's theorem, we have shown there is a $k$ such that:
    \begin{align*}
        \mbf{x}^*(L_{k-1})
        = |N(A)|
        < |A|
        \le \mbf{x}(L_k)
    \end{align*}
    To construct $\wt{\mbf{x}}$, we apply \Cref{lem:mmo:exchange}.
    The above inequality in conjunction with \Cref{lem:mmo:exchange} implies that there is $i\in L_k$ and $i^*\in L_{k-1}^c$ such that $\wt{\mbf{x}} = \mbf{x} - \chi_{i} + \chi_{i^*} \in \base(P)$.

    We complete the contradiction by showing $\mbf{x} + \bell \npreceq \wt{\mbf{x}} + \bell$ using the thresholding characterization of majorization.
    Using the notation $\mbf{v} \define \mbf{x} + \bell$ and $\wt{\mbf{v}} \define \wt{\mbf{x}} + \bell$ along with threshold $\gamma \define k - 1$ gives:
    \begin{align}\label{eq:mmo:bij_maj}
        &\sum_{i'\in [n]}\pp{\wt{\mbf{v}}(i') - \gamma} \nonumber \\
        &= \pp{\wt{\mbf{v}}(i) - \gamma} + \pp{\wt{\mbf{v}}(i^*) - \gamma} + \sum_{i'\in [n]\setminus \{i, i^*\}}\pp{\wt{\mbf{v}}(i') - \gamma} \nonumber \\
        &= \pp{(\mbf{v}(i) - 1) - \gamma} + \pp{(\mbf{v}(i^*) + 1) - \gamma} + \sum_{i'\in [n]\setminus \{i, i^*\}}\pp{\mbf{v}(i') - \gamma} \nonumber \\
        &= \pp{\mbf{v}(i) - \gamma} - 1 + \pp{\mbf{v}(i^*) - \gamma} + \sum_{i'\in [n]\setminus \{i, i^*\}}\pp{\mbf{v}(i') - \gamma} \\
        &= \sum_{j\in [n]} \pp{\mbf{v}(i') - \gamma} - 1 \nonumber
    \end{align}
    \Cref{eq:mmo:bij_maj} holds since $\mbf{v}(i) > \gamma$ and $\mbf{v}(i^*) < \gamma$.
    Notice that $\mbf{x} + \bell$ is strictly greater that $\wt{\mbf{x}} + \bell$ on threshold $k-1$, so we get the contradiction $\mbf{x} + \bell \npreceq \wt{\mbf{x}} + \bell$.
\end{proof}

\section{AI Disclosure}\label{app:sec:ai}
We used Google Gemini 3 to assist with editing drafts.
The tool materially affected the abstract and related work (\Cref{subsec:intro:related_work}).
Drafts of each section were provided to Gemini with the request to provide grammatical support and streamlining the flow of writing.
The suggestions of Gemini were implemented in part, as the authors saw fit.

In addition to editorial support, the authors used  Claude (Sonnet 4.6) to assist in the proof of \Cref{lem:mmo:maj_min_bij}.
The use of Hall's theorem to establish $\mbf{x}(L_{k-1}) < \mbf{x}(L_k)$ was derived from a proof sketch by Claude.
The author identified the exchange lemma \ref{lem:mmo:exchange} as a sufficient condition for \Cref{lem:mmo:maj_min_bij}; Claude was unable to assist in showing this exchange lemma, so the authors queried (placeholder for person's name), who outlined the provided proof.
 
The authors verified the correctness and originality of all content generated by AI to the best of their ability, including references.

\end{document}